\documentclass{article}
\usepackage{amsmath,amsthm,amssymb}
\usepackage{microtype}

\usepackage{hyperref}
\usepackage{fullpage}

\usepackage{tikz}
\usetikzlibrary{shapes.symbols,arrows,snakes}
\usetikzlibrary{matrix,calc,intersections,patterns}

\usepackage{enumerate}
\usepackage{algorithm,algorithmic}

\newcommand{\INITIALLY}{\REQUIRE{}}
\newcommand{\ROUND}{\ENSURE{}}

\newcommand{\nop}[1]{}

\newtheorem{thm}{Theorem}

\newtheorem{lem}[thm]{Lemma}

\DeclareMathOperator{\diam}{diam}

\DeclareMathOperator*{\argmax}{arg\,max}

\newcommand{\N}{\mathcal{N}}

\newcommand{\IR}{\mathbb{R}}

\renewcommand{\leq}{\leqslant}
\renewcommand{\ge}{\geqslant}
\renewcommand{\geq}{\geqslant}

\usepackage{graphicx}

\DeclareMathOperator\Rcv{Rcv}

\title{Fast Multidimensional Asymptotic and Approximate Consensus}
\author{Matthias F\"ugger\\
  CNRS, LSV, ENS Paris-Saclay, Inria\\
  {\tt mfuegger@lsv.fr}\\
  \and
  Thomas Nowak\\
  Universit\'e Paris-Sud\\
  {\tt thomas.nowak@lri.fr}}

\title{Fast Multidimensional Asymptotic and Approximate Consensus}

\date{}

\begin{document}

\maketitle

\begin{abstract}
We study the problems of asymptotic and approximate consensus in which 
agents have to get their values arbitrarily close to each others'
inside the convex hull of initial values, either without or with an explicit
decision by the agents.
In particular, we are concerned with the case of multidimensional data, i.e.,
the agents' values are $d$-dimensional vectors.
We introduce two new algorithms for dynamic networks,
subsuming classical failure models like asynchronous message passing systems with Byzantine agents.
The algorithms are the first to have a contraction rate and time complexity independent of the dimension~$d$.
In particular, we improve the time complexity
from the previously fastest approximate consensus algorithm in asynchronous
message passing systems with Byzantine faults 
by Mendes et al. [Distrib.\ Comput.~28]
from $\Omega\!\left( d \log\frac{d\Delta}{\varepsilon} \right)$
to $O\!\left( \log\frac{\Delta}{\varepsilon} \right)$,
where~$\Delta$ is the initial and~$\varepsilon$ is the terminal diameter
of the set of vectors of correct agents.
\end{abstract}

\section{Introduction}

The problem of one-dimensional asymptotic consensus requires a system of agents, starting from potentially different
  initial real values, to repeatedly set their local output variables such that all outputs
  converge to a common value within the convex hull of the inputs.
This problem has been studied in distributed control theory both from a theoretical perspective~\cite{Cha11,Mor05,CMA08b,AB06} and
  in the context of robot gathering on a line~\cite{BPGS10} and clock synchronization~\cite{OM04,LR06}.
Extensions of the problem to multidimensional values naturally arise in the context of
  robot gathering on a plane or three-dimensional space~\cite{CFPS03},
  as subroutines in formation forming~\cite{Cha11}, and distributed
  optimization~\cite{boyd2004convex}, among others.

The related problem of approximate consensus, also called approximate agreement, requires the agents to
  eventually decide, i.e., not change their output variables.
Additionally all output variables must be within a predefined $\varepsilon > 0$ distance of each other and lie within the
  convex hull of the inputs.
There is a large body of work on approximate consensus in distributed computing devoted to solvability and optimality
  of time complexity~\cite{DLPSW86,Fek90} and applications in clock synchronization; see e.g.~\cite{WL88,schneider1987understanding}.

Both problems were studied under different assumptions on the underlying communication between agents and
  their computational strength, including fully connected asynchronous message passing with Byzantine agents~\cite{WL88,DLPSW86}
  and communication in rounds by message passing in dynamic communication networks~\cite{Mor05,Cha11}.
In~\cite{CBFN15,CBFN16} Charron-Bost et al.\ analyzed solvability of asymptotic consensus and approximate consensus in
  dynamic networks with round-wise message passing defined by {\em network models}: a network model is a set of directed communication
  graphs, each of which specifies successful reception of broadcast messages; see Section~\ref{sec:dynmodel} for a formal definition.
Solving asymptotic consensus in such a model requires to fulfill the specification of asymptotic consensus in {\em any\/} sequence
  of communication graphs from the model.
Charron-Bost et al.\ showed that in these highly dynamic networks, asymptotic consensus and approximate consensus are solvable in
  a network model if and only if each of its graphs contains a spanning rooted tree.
An interesting class of network models are those that contain only {\em non-split\/} communication graphs, i.e.,
  communication graphs where each pair of nodes has a common incoming neighbor.
Several classical fault-model were shown to be instances of non-split models~\cite{CBFN15}, among them asynchronous message passing
  systems with omissions.

Recently the multidimensional version of approximate consensus received
attention.  
Mendes et al. \cite{MHVG15} were the first to present algorithms that solve
approximate consensus in Byzantine message passing systems
for $d$-dimensional real vectors.
Their algorithms, Mendes–Herlihy and Vaidya–Garg, are based on the repeated construction of so called safe areas of received vectors
  to constraint influence of values sent by Byzantine agents, followed
  by an update step, ensuring that the new output values are in the safe area.
They showed that the diameter of output values contracts by at least $1/2$ in each
  dimension every~$d$ rounds in the Mendes–Herlihy algorithm,
  and the diameter of the output values contracts by at least $1-1/n$ every round in the Vaidya–Garg algorithm,
  where~$n$ is the number of agents.
The latter bound assumes $f=0$ Byzantine failures and slightly worsens for $f>0$.
In terms of contraction rates as introduced in~\cite{FNS18:podc} (see Section~\ref{sec:metrics} for a definition)
  of the respective non-terminating algorithms for asymptotic consensus, they thus obtain
  upper bounds of $\sqrt[d]{1/2}$ and $1-1/n$.
Note that the Mendes–Herlihy algorithm has contraction rate depending only on $d$ but
  requires an a priori common coordinate system, while the Vaidya–Garg algorithm is coordinate-free but has a contraction
  rate depending on $n$.

Charron-Bost et al.~\cite{CBFN16:centroid} analyzed convergence of the Centroid algorithm where agents repeatedly update their
  position to the centroid of the convex hull of received vectors.
The algorithm is coordinate-free and has a contraction rate of $d/(d+1)$, independent of $n$.  
Local time complexity of determining the centroid was shown to be \#P-complete~\cite{Rad07} while polynomial in $n$ for fixed $d$.

The contraction rate of the Centroid algorithm is always smaller or equal to
that of the Mendes-Herlihy algorithm, though both converge to~$1$ at the 
same speed. 
More precisely,
\[
\lim_{d\to\infty}
\frac{\left\lvert 1 - \sqrt[d]{\frac{1}{2}}\right\rvert}
{\left\rvert 1 - \frac{d}{d+1}\right\rvert}
=
\log 2
\enspace,
\]
which implies
$\left\lvert 1 - \sqrt[d]{\frac{1}{2}}\right\rvert
=
\Theta
\left(
\left\rvert 1 - \frac{d}{d+1}\right\rvert
\right)$.

\subsection{Contribution}

In this work we present two new algorithms that are coordinate-free: the MidExtremes and the ApproachExtreme algorithm, and study their behavior in dynamic networks.
Both algorithms are coordinate-free, operate in rounds, and are shown to solve asymptotic agreement in non-split network models.
Terminating variants of them are shown to solve approximate agreement in non-split network models.

As a main result we prove that their convergence rate is independent of network size $n$ and dimension $d$ of the initial values.
For MidExtremes we obtain an upper bound on the contraction rate of $\sqrt{7/8}$ and for ApproachExtreme of $\sqrt{31/32}$.

Due to the fact that classical failure models like asynchronous message passing with Byzantine agents possess corresponding
network models, our results directly yield improved algorithms for the latter failure models:
In particular, we improve the time complexity from the previously fastest approximate consensus algorithm in asynchronous
message passing systems with Byzantine faults, the Mendes–Herlihy algorithm,
from $\Omega\!\left( d \log\frac{d\Delta}{\varepsilon} \right)$
to $O\!\left( \log\frac{\Delta}{\varepsilon} \right)$,
where~$\Delta$ is the initial and~$\varepsilon$ is the terminal diameter
of the set of vectors of correct agents.
Note that our algorithms share the benefit of being coordinate-free with the Vaidya–Garg algorithm presented in the same work.

Table~\ref{tab:sum} summarizes our results and the algorithms discussed above for asymptotic and approximate consensus.
The table compares the new algorithms MidExtremes and ApproachExtreme to the Centroid, Mendes–Herlihy (MH), and Vaidya–Garg (VG) algorithms
with respect to their local time complexity per agent and round and an upper bound on their contraction rate.
A lower bound of $1/2$ on the contraction rate is due to F\"ugger et al.\ \cite{FNS18:podc}.
 
\begin{table}
\centering
\small  
{\renewcommand{\arraystretch}{1.4}
\begin{tabular}{| c | c | c | c | c | c | c |}
\hline
&  MidExtremes &  ApproachExtreme & Centroid & MH & VG\\
\hline
contraction rate & $\sqrt{\frac{7}{8}}$\textsuperscript{*} & $\sqrt{\frac{31}{32}}$\textsuperscript{*} & $\frac{d}{d+1}$\ & $\sqrt[d]{\frac{1}{2}}$ & $1-\frac{1}{n}$\\
\hline
local {\tt TIME} & $O(n^2d)$ & $O(nd)$ & \#P-hard & $O(nd)$ & $O(nd)$ \\
\hline
coordinate-free & yes & yes & yes & no & yes\\
\hline
\end{tabular}
}
\caption{Comparison of local time complexity and contraction rates.
  Entries marked with an \textsuperscript{*} are new results in this paper.
}
\label{tab:sum}
\end{table}

The Mendes-Herlihy algorithm has a smaller contraction rate than the 
MidExtremes algorithm whenever $d\leq 10$;
the Centroid algorithm whenever $d\leq 14$.
The Centroid algorithm is hence the currently fastest known algorithm for 
dimensions $3\leq d\leq 14$.
For dimensions $d=1$ and $d=2$, the componentwise MidPoint algorithm has an
optimal contraction rate of~$1/2$\ \cite{CBFN16:centroid}.
Note that the MidExtremes algorithm is equivalent to the componentwise MidPoint
  algorithm for dimension $d=1$.
For $d\geq 15$, the MidExtremes algorithm is the currently fastest known
algorithm.

We finally note that all our results hold for the class of inner product spaces
and are not restricted to
the finite-dimensional Euclidean spaces $\IR^d$, in contrast to previous
work.
For example, this includes the set of square-integrable functions on a real interval.

\section{Model and Problem}

We fix some vector space~$V$ with an inner product 
$\langle \cdot , \cdot \rangle : V\times V \to \IR$
and the norm $\lVert x\rVert = \sqrt{\langle x, x\rangle}$.
The prototypical finite-dimensional example is
$V = \IR^d$
with the usual inner product
and the Euclidean norm.
The diameter of set $A\subseteq V$ is denoted by
  $\displaystyle\diam(A) = \sup_{x,y\in A} \lVert x-y \rVert$.
For an $n$-tuple $x=(x_1,\dots,x_n)\in V^n$ of vectors in~$V$,
we write
  $\Delta(x) = \diam\big(\{x_1,\dots, x_n\}\big)$.

\subsection{Dynamic Network Model}
\label{sec:dynmodel}

We consider a distributed system of~$n$ agents that communicate in
  rounds via message passing, like in the Heard-Of model~\cite{CS09}.
In each round, each agent~$i$, broadcasts a message
  based on its local state, receives some messages, and 
  then updates its local state
  based on the received messages and its local state.
Rounds are communication closed: agents only receive
  messages sent in the same round.

In each round $t \ge 0$, messages are delivered according to the
  {\em directed communication graph\/} $G_t$ for round~$t$:
  the message broadcast by~$i$ in round~$t$ is
  received by~$j$ if and only if the directed edge $(i,j)$ is in~$G_t$.
Agents always receive their own messages, i.e., $(i,i) \in G_t$.
A {\em communication pattern\/} is an infinite sequence $G_1, G_2, \dots$ of
communication graphs.
A (deterministic) {\em algorithm\/} specifies, for each agent~$i$, the local
state space of~$i$,
  the set of initial states of~$i$, 
  the sending function for which message to broadcast, and the state transition
  function.  
For asymptotic consensus,
each agent~$i$'s local state necessarily contains a variable $y_i \in V$,
which initially holds $i$'s input value and 
is then used as its output variable.
We require the that there is an initial state with initial value~$v$ for all 
vectors $v\in V$.
A {\em configuration\/} is an $n$-tuple of local states.
It is called initial if all local states are initial.  
The {\em execution\/} of an algorithm from initial configuration $C_0$ induced
by communication pattern
  $G_1, G_2, \dots$ is the unique sequence $C_0, G_1, C_1, G_2, C_2, \dots$
  alternating between configurations and communication graphs
  where $C_{t}$ is the configuration obtained by delivering messages in round
  $t$ according to
  communication graph~$G_t$, and applying the sending and local transition
  functions to the local states
  in $C_{t-1}$ according to the algorithm.
For a fixed execution and a local variable~$z$ of the algorithm,
we denote by $z_i(t)$ its value at~$i$ at the end of round~$t$, i.e.,
in configuration~$C_t$.
In particular, $y_i(t)$ is the value of $y_i$ in $C_t$.
We write
$y(t) = \big( y_1(t), \dots, y_n(t) \big)$
for the collection of the $y_i(t)$.

A specific class of algorithms for asymptotic consensus 
are the so-called {\em convex combination}, or {\em averaging},
algorithm, which only ever update the value of~$y_i$ inside the convex hull
of~$y_j$ it received from other agents $j$ in the current round.
Many algorithms in the literature belong to this class, as do ours.

Following~\cite{CBFN15}, we study the behavior of algorithms for communication patterns
  from a {\em network model}, i.e., a non-empty set of
  communication patterns: a communication pattern is from network model $\N$ if all its
  communication graphs are in $\N$.
We will later on show that such an analysis also allows to prove new performance bounds for
  more classical fault-models like asynchronous message passing systems with Byzantine agents. 

An interesting class of network models are so called {\em non-split\/} models, i.e., those
  that contain only non-split communication graphs: a communication graph is non-split if
  every pair of nodes has a common in-neighbor.
Charron-Bost et al.\ \cite{CBFN15} showed that asymptotic and approximate consensus is
  solvable efficiently in these network models in the case of one dimensional values.
They further showed that: (i) In the weakest (i.e., largest) network model in which asymptotic
  and approximate consensus are solvable, the network model of all communication graphs that
  contain a rooted spanning tree, one can simulate non-split communication graphs.
  (ii) Classical failure models like link failures as considered in~\cite{SW89} and asynchronous message passing
  systems with crash failures have non-split interpretations.
Indeed we will make use of such a reduction from non-split network models to asynchronous message passing
  systems with Byzantine failures in Section~\ref{sec:asyncbyz}.

\subsection{Problem Formulation}

An algorithm {\em solves the asymptotic consensus problem\/} in a network model $\N$ 
  if the following holds for every execution with a communication pattern from $\N$: 
\begin{itemize}
  \item{\em Convergence.\/} Each sequence $\big(y_i (t)\big)_{t\geq 0}$ converges.
    
  \item{\em Agreement.\/} If $y_i(t)$ and $y_j(t)$ converge, then they have a common limit.
    
  \item{\em Validity.\/} If $y_i(t)$ converges, then its limit is in the convex hull of the
    initial values $y_1(0), \dots, y_n(0)$.
	
\end{itemize}

For the deciding version, the {\em approximate consensus\/} problem (see, e.g., \cite{Lyn96}),
  we augment the local state of~$i$ with a variable~$d_i$ initialized to~$\bot$.
Agent~$i$ is allowed to set~$d_i$ to some value $v\neq \bot$ only once, in which case
  we say that~$i$ {\em decides}~$v$.
In addition to the initial values~$y_i(0)$, agents initially receive the error
  tolerance~$\varepsilon$ and an upper bound~$\Delta$ on the maximum distance of
  initial values.
An algorithm {\em solves approximate consensus\/} in $\N$ if for all
  $\varepsilon >0$ and all~$\Delta$, 
  each execution with a communication pattern in $\N$ with initial diameter at most~$\Delta$ satisfies:
\begin{itemize}
\item{\em Termination.\/} Each agent eventually decides.
	
\item{\em $\varepsilon$-Agreement.\/} If agents $i$ and $j$ decide $d_i$ and
    $d_j$, respectively, then $\lVert d_i - d_j' \lVert \,\leq \varepsilon$.
  
\item{\em Validity.\/} If agent $i$ decides $d_i$, then $d_i$ is in the convex hull of 
  initial values $y_1(0), \dots, y_n(0)$.
	
\end{itemize}		

\subsection{Performance Metrics}\label{sec:metrics}

The {\em valency\/} of a configuration~$C$,
denoted by $Y^*(C)$,
is defined as the set of limits
of the values~$y_i$ in executions that include 
configuration~$C$~\cite{FNS18:podc}.
If the execution is clear from the context, we abbreviate
$Y^*(t) = Y^*(C_t)$.
The {\em contraction rate\/} of an execution is defined as
\begin{equation*}
\limsup_{t\to\infty}
\sqrt[t]{\diam\big(Y^*(t)\big)}
\enspace.
\end{equation*}
The contraction rate of an algorithm in a network model is the supremum
of the contraction rates of its executions.
For convex combination algorithms,
the contraction rate is always upper-bounded by its
{\em convergence rate}, that is,
\begin{equation*}
\limsup_{t\to\infty}
\sqrt[t]{\diam\big(Y^*(t)\big)}
\leq
\limsup_{t\to\infty}
\sqrt[t]{\Delta\big(y(t)\big)}
\enspace.
\end{equation*}
We define
the {\em round-by-round convergence rates\/} by
\[
    c(t) = \frac{\Delta\big(y(t)\big)}{\Delta\big(y(t-1)\big)}
\]
for a given execution and a given round $t\geq 1$.
Clearly, an algorithm that guarantees a round-by-round convergence rate
of $c(t)\leq \beta$ also guarantees a convergence rate of at most~$\beta$.
Since both of our algorithms are convex combination algorithms, all our 
upper bounds on the round-by-round convergence rates are also upper bounds for 
the contraction rates.

The {\em convergence time\/} of a given execution measures the time from which
on all values are guaranteed to be in an~$\varepsilon$ of each other.
Formally, it is the function defined as
\begin{equation*}
T(\varepsilon)
=
\min
\big\{
t\geq 0
\mid
\forall \tau\geq t\colon
\Delta\big(x(\tau)\big) \leq \varepsilon
\big\}
\enspace.
\end{equation*}
In an execution that satisfies $c(t) \leq \beta$ for all $t\geq 1$, we
have the bound
$T(\varepsilon) \leq \left\lceil \log_{1/\beta} \frac{\Delta}{\varepsilon}\right\rceil$
on the convergence time,
where $\Delta = \Delta\big(y(0)\big)$ is the diameter of the set of initial
values.
 
\section{Algorithms}
\label{sec:algorithms}

In this section, we introduce two new algorithms for solving asymptotic
and approximate consensus in arbitrary inner product spaces with constant 
contraction rates.
We present our algorithms and prove their correctness and bounds on
their performance in non-split networks models.
While we believe that this framework is the one in which our arguments
are clearest,
our results can be extended to a number of other models whose
underlying communication graphs turn out to be, in fact, non-split.
The following is a selection of these models:

\begin{itemize}
\item Rooted network models: This is the largest class of network models
in which asymptotic and approximate consensus are solvable~\cite{CBFN15}.
A network model is rooted if all its communication graphs include a
directed rooted spanning tree, though not necessarily the same in all graphs.
Although not every such communication graph is non-split, 
Charron-Bost et al.~\cite{CBFN15}
showed that the cumulative communication graph over $n-1$ rounds in
a rooted network model is always non-split.
In such network models, one can use amortized versions~\cite{CBFN16}
of the algorithms, which operate in macro-rounds of $n-1$ rounds each.
If an algorithm has a contraction rate~$\beta$ in non-split network models,
then its amortized version has contraction rate $\sqrt[n-1]{\beta}$
in rooted network models.
The amortized versions of our algorithms thus have contraction rates independent
of the dimension of the data.
\item Omission faults: In the omission fault model studied by Santoro and
Widmayer~\cite{SW89}, the adversary can delete up to~$t$ messages from
a fully connected communication graph each round.
If $t\leq n-1$, then all communication graphs are non-split.
If $t \leq 2n-3$, then all communication graphs are rooted~\cite{CBFN15}.
Our algorithms are hence applicable in both these cases and have contraction
rates independent of the dimension.
\item Asynchronous message passing with crash faults:
Building asynchronous rounds atop of asynchronous message passing by waiting
for $n-f$ messages in each round, the resulting
communication graphs are non-split as long as the number~$f$ of possible 
crashes is strictly smaller than $n/2$.
We hence get a constant contraction rate using our algorithms also in this 
model.
For $f \geq n/2$, a partition argument shows that neither asymptotic nor 
approximate consensus are solvable.
\item Asynchronous message passing with Byzantine faults:
Mendes et al.~\cite{MHVG15} showed that approximate consensus is solvable
in asynchronous message passing systems with~$f$ Byzantine faults if and
only if $n > (d+2)f$ where~$d$ is the dimension of the data.
The algorithms they presented construct a round structure whose 
communication graphs turn out to be non-split.
Since the construction is not straightforward, we postpone the discussion
of our algorithms in this model to Section~\ref{sec:asyncbyz}.
\end{itemize}

\subsection{Non-split Network Models}
  
We now present our two new algorithms, MidExtremes and ApproachExtreme.
Both operate in the following simple round structure:
broadcast the current value~$y_i$ and then update it to a new value depending
on the set $\Rcv_i$ of values~$y_j$ received from other agents in the current
round.
Both of them only need to calculate distances between values and 
form the midpoint between two values.
In particular, we do not need to make any assumption on the dimension of
the space of possible values for implementing the algorithms.
We only need a distance and an affine structure, for calculating the midpoint.
Our correctness proofs, however, rely on the fact that the distance function is
a norm induced by an inner product.

Note that, although we present algorithms for asymptotic consensus,
combined with our upper bounds on the convergence time,
one can easily deduce versions for approximate consensus by having
the agents decide after the upper bound.
Our upper bounds only depend on the precision parameter~$\varepsilon$
and (an upper bound on) the initial diameter~$\Delta$.
While upper bounds on the initial diameter cannot be deduced during
execution in general non-split network models, it can be done in specific
models, like asynchronous message passing with Byzantine 
faults~\cite{MHVG15}.
Otherwise, we need to assume an a priori known bound on the initial
diameter to solve approximate consensus.

The algorithm MidExtremes, which is shown in Algorithm~\ref{alg:midextremes},
updates its value~$y_i$ to the midpoint of a pair of extremal points of 
$\Rcv_i$ that realizes its diameter.
In the worst case, it thus has to compare the distances of $\Theta(n^2)$ pairs
of values.
For the specific case of Euclidean spaces $V = \IR^d$ stored in a
component-wise representation, this amounts to $O(n^2d)$ local 
scalar operations for each agent in each round.
  
\begin{algorithm}[ht]
\begin{algorithmic}[1]
\INITIALLY{}
\STATE $y_i$ is the initial value in $V$ 
\ROUND{}
\STATE broadcast $y_i$
\STATE $\Rcv_i \gets$ set of received values
\STATE 
$\displaystyle (a,b) \gets \argmax_{ (a,b) \in Rcv_i^2} \lVert a - b\rVert$
\STATE 
$\displaystyle y_i \gets \frac{a+b}{2}$
\end{algorithmic}
\caption{Asymptotic consensus algorithm MidExtremes for agent $i$}
\label{alg:midextremes}
\end{algorithm}

It turns out that we can show a round-by-round convergence rate of the
MidExtremes
algorithm  independent of the dimension or the number of agents,
namely $\sqrt{7/8}$.
For the specific case of values from the real line $V = \IR$,
it reduces to the MidPoint algorithm~\cite{CBFN16}, whose contraction rate
of~$1/2$
is known to be optimal~\cite{FNS18:podc}.

\begin{thm}\label{thm:midextremes}
In any non-split network model with values from any inner product space,
the MidExtremes algorithm guarantees
a round-by-round convergence rate of
$c(t) \leq \sqrt{{7}/{8}}$ for all rounds $t\ge 1$.
Its convergence time is at most
$T(\varepsilon) = \left\lceil\log_{\sqrt{8/7}}
\frac{\Delta}{\varepsilon}\right\rceil$
where~$\Delta$ is the diameter of the set of initial values.

In the particular case of values from the real line, it
guarantees a round-by-round convergence rate of $c(t) \leq 1/2$
and a convergence time of
$T(\varepsilon) = \left\lceil\log_{2}
\frac{\Delta}{\varepsilon}\right\rceil$.
\end{thm}

The second algorithm we present is called ApproachExtreme and shown in
Algorithm~\ref{alg:approachextreme}.
It updates its value~$y_i$ to the midpoint of the current value of~$y_i$ 
and the value in $\Rcv_i$ that is the farthest from it.
While having the benefit of only having to compare $O(n)$ distances,
and hence doing $O(nd)$ local scalar operations
for each agent in each round
in the case of $V=\IR^d$ with component-wise representation,
the ApproachExtreme algorithm also only has to measure distances from
its current value to other agents' values; never the distance of two other
agents' values.
This can be helpful for agents embedded into the vector space~$V$ that can
measure the distance from itself
to another agent, but not necessarily the distance between two other agents.

\begin{algorithm}[ht]
\begin{algorithmic}[1]
\INITIALLY{}
\STATE $y_i$ is the initial value in $V$ 
\ROUND{}
\STATE broadcast $y_i$
\STATE $\Rcv_i \gets$ set of received values
\STATE 
$\displaystyle b \gets \argmax_{b \in Rcv_i} \lVert y_i - b\rVert$
\STATE $\displaystyle y_i \gets \frac{y_i+b}{2}$
\end{algorithmic}
\caption{Asymptotic consensus algorithm ApproachExtreme for agent $i$}
\label{alg:approachextreme}
\end{algorithm}

The ApproachExtreme algorithm admits an upper bound of $\sqrt{31/32}$
on its round-by-round convergence rate, which is worse than the $\sqrt{7/8}$
of the MidExtremes algorithm.
For the case of the real line $V = \IR$, we can show a round-by-round
convergence rate of~$3/4$, however.

\begin{thm}\label{thm:approachextreme}
In any non-split network model with values from any inner product space,
the ApproachExtreme algorithm guarantees
a round-by-round convergence rate of
$c(t) \leq \sqrt{\frac{31}{32}}$ for all rounds $t\ge 1$.
Its convergence time is at most
$T(\varepsilon) = \left\lceil\log_{\sqrt{32/31}}
\frac{\Delta}{\varepsilon}\right\rceil$
where~$\Delta$ is the diameter of the set of initial values.

In the particular case of values from the real line, it
guarantees a round-by-round convergence rate of $c(t) \leq 3/4$
and a convergence time of
$T(\varepsilon) = \left\lceil\log_{4/3}
\frac{\Delta}{\varepsilon}\right\rceil$.
\end{thm}

\subsection{Asynchronous Byzantine Message Passing}
\label{sec:asyncbyz}

We now show how to adapt algorithm MidExtremes to the case of
asynchronous message passing systems with at most~$f$ Byzantine agents.
The algorithm proceeds in the same asynchronous round structure
and safe area calculation used by Mendes et al.~\cite{MHVG15}
whenever approximate consensus is solvable, i.e., when $n > (d+2)f$.
Plugging in the MidExtremes algorithm, we achieve a round-by-round convergence
rate and round complexity independent of the dimension~$d$.

More specifically, our algorithm has
a round complexity of $O\!\left(\log \frac{\Delta}{\varepsilon}\right)$,
which leads to
a message complexity of $O\!\left(n^2 \log \frac{\Delta}{\varepsilon}\right)$
where~$\Delta$ is the maximum Euclidean distance of initial vectors of {correct}
agents.
In contrast, the Mendes-Herlihy algorithm has a worst-case round complexity of
$\Omega\!\left( d\log \frac{d\Delta}{\varepsilon}\right)$
and a worst-case message complexity of
$\Omega\!\left( n^2d\log \frac{d\Delta}{\varepsilon}\right)$.
We are thus able to get rid of all terms depending on the dimension~$d$.

After an initial round estimating the initial diameter of the system,
the Mendes-Herlihy algorithm has each agent~$i$ repeat the following steps in
each coordinate
$k\in \{1,2,\dots,d\}$ for 
$\Theta\!\left( \log \frac{d\Delta}{\varepsilon} \right)$ rounds:

\begin{enumerate}
\item Collect a multiset~$V_i$ of agents' vectors such that every
intersection $V_i \cap V_j$ has at least $n-f$ elements via reliable broadcast
and the witness technique~\cite{AAD05}.
\item Calculate the safe area~$S_i$ as the intersection of the convex hulls
of all sub-multisets of~$V_i$ of size $\lvert V_i \rvert - f$.
The safe area is guaranteed to be a subset of the convex hull of vectors
of correct agents.
Helly's theorem~\cite{DGK63} can be used to show that every intersection
$S_i \cap S_j$ of safe areas is nonempty.

\item Update the vector~$y_i$ to be in the safe area~$S_i$ and have its
$k$\textsuperscript{th}
coordinate equal to the midpoint of the set of $k$\textsuperscript{th}
coordinates in~$S_i$.
\end{enumerate}

The fact that safe areas have nonempty pairwise intersections guarantees
that the diameter in the $k$\textsuperscript{th} coordinate
\begin{equation*}
\delta_k(t) = \max_{i,j\text{ correct}} 
\left\lvert y_i^{(k)}(t) - y_j^{(k)}(t) \right\rvert
\end{equation*}
at the end of round~$t$ fulfills 
$\delta_k(t) \leq \delta_k(t-1)/2$
if round~$t$ considers coordinate~$k$.
The choice of the number of rounds for each coordinate guarantees that
we have $\delta_k(t) \leq \varepsilon / \sqrt{d}$ after the last round for
coordinate~$k$.
This in turn makes sure that the Euclidean diameter of the set of vectors of
correct agents after all of the
$\Theta\!\left( d \log\frac{d\Delta}{\varepsilon} \right)$
rounds is at most~$\varepsilon$.

The article of Mendes et al.~\cite{MHVG15} describes a second
algorithm, the Vaidya-Garg algorithm, which
replaces steps~2 and~3 by updating~$y_i$ to the non-weighted average of
arbitrarily chosen points in the safe areas of all sub-multisets of~$V_i$ of
size $n-f$.
Another difference to the Mendes-Herlihy algorithm is that it repeats the steps
not several times for every dimension, but for
$\Theta\!\left( n^{f+1} \log \frac{d\Delta}{\varepsilon}\right)$
rounds in total.
The Vaidya-Garg algorithm comes with the advantage of not having to do the
calculations to find a midpoint for the $k$\textsuperscript{th} coordinate
while remaining inside the safe area, but also comes with the cost of a 
convergence rate and a round complexity that depends on the number of
agents.

The algorithm we propose has the same structure as the Mendes-Herlihy
algorithm, but with the difference that 
we replace step~3 by updating vector~$y_i$ to the midpoint of two points
that realize the Euclidean diameter of the safe area~$S_i$.
According to our results in Section~\ref{sec:perf:midextremes}, 
the Euclidean diameter
\begin{equation*}
\delta(t) = \max_{i,j\text{ correct}} \big\lVert y_i(t) - y_j(t) \big\rVert
\end{equation*}
of the set of vectors of correct agents at the end of round~$t$
satisfies
\begin{equation*}
\delta(t) \leq \sqrt{\frac{7}{8}} \delta(t-1)
\enspace.
\end{equation*}
This means that we have $\delta(T) \leq \varepsilon$ after
\begin{equation*}
T(\varepsilon) =
\left\lceil
\log_{\sqrt{8/7}} \frac{\Delta}{\varepsilon}
\right\rceil
\end{equation*}
rounds.

\section{Performance Bounds}\label{sec:perf}

We next show upper bounds on the round-by-round convergence rate for algorithms
  MidExtremes (Theorem~\ref{thm:midextremes}) and ApproachExtreme (Theorem~\ref{thm:approachextreme})
  in non-split network models.

\subsection{Bounds for MidExtremes}\label{sec:perf:midextremes}

For dimension $1$, MidExtremes is equivalent to the MidPoint Algorithm.
We hence already know that $c(t) \leq \frac{1}{2}$ from~\cite{CBFN16}, proving
  the case of the real line in Theorem~\ref{thm:midextremes}.

For the case of higher dimensions we will show that $c(t) \leq \sqrt{\frac{7}{8}}$ holds.
The proof idea is as follows:
For a round $t \ge 1$, we consider two agents $i,j$
  whose distance realizes $\Delta(y(t))$.
By the algorithm we know that both agents set their $y_i(t)$ and $y_j(t)$ 
  according to $y_i(t) = m = (a+b)/2$ and $y_j(y) = m' = (a'+b')/2$,
  where $a,b$ are the extreme points received by agents $i$ in round $t$ and
  $a',b'$ are the extreme points received by agents $j$ in the same round.
%Assuming for the moment that $V = \IR^d$, $a,b,a',c'$ are from $\IR^d$, 
All four points must lie within a common subspace
  of dimension $3$, and form the vertices of a tetrahedron as depicted in Figure~\ref{fig:tetrahedron}.

\tikzset{
  schraffiert/.style n args={2}{pattern=#2,pattern color=#1},
  schraffiert/.default=black
}

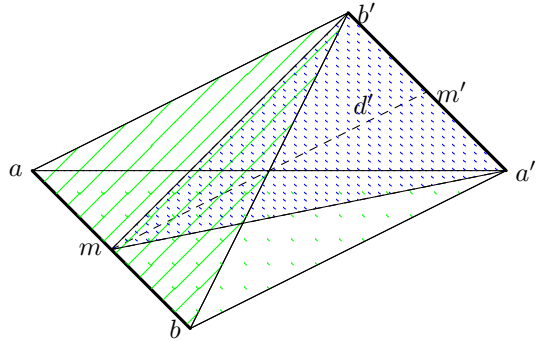
\begin{figure}[b]
\centering
\begin{tikzpicture}[line join = round, line cap = round, scale=2.1]
  \pgfdeclarepatternformonly{mypattern1}%
   {\pgfqpoint{-1pt}{-1pt}}%
   {\pgfqpoint{10pt}{10pt}}%
   {\pgfqpoint{9pt}{9pt}}%
   {
     \pgfsetlinewidth{0.4pt}
     \pgfpathmoveto{\pgfqpoint{0pt}{0pt}}
     \pgfpathlineto{\pgfqpoint{9.1pt}{9.1pt}}
     \pgfusepath{stroke}
   }
   \pgfdeclarepatternformonly{mypattern2}%
   {\pgfqpoint{-1pt}{-1pt}}%
   {\pgfqpoint{10pt}{10pt}}%
   {\pgfqpoint{9pt}{9pt}}%
   {
     \pgfsetlinewidth{0.4pt}
     \pgfpathmoveto{\pgfqpoint{0pt}{0pt}}
     \pgfpathlineto{\pgfqpoint{-9.1pt}{9.1pt}}
     \pgfusepath{stroke}
   }
   \pgfdeclarepatternformonly{mypattern3}%
   {\pgfqpoint{-1pt}{-1pt}}%
   {\pgfqpoint{10pt}{10pt}}%
   {\pgfqpoint{4pt}{3pt}}%
   {
     \pgfsetlinewidth{0.4pt}
     \pgfpathmoveto{\pgfqpoint{0pt}{0pt}}
     \pgfpathlineto{\pgfqpoint{-4.1pt}{3.1pt}}
     \pgfusepath{stroke}
   }

  \coordinate[label=left:$a$] (A) at (0,0);
  \coordinate[label=left:$b$] (B) at (1,-1);
  \coordinate[label=left:$m$] (M) at (0.5,-0.5);
  
  \coordinate[label=right:$a'$] (As) at (3,0);
  \coordinate[label=right:$b'$] (Bs) at (2,1);
  \coordinate[label=right:$m'$] (Ms) at (2.5,0.5);

  \draw[-,very thick] (A) -- (B) node[pos=0.25,left] {} node[pos=0.75,left] {};
  \draw[-] (B) -- (As) node[pos=0.5,right] {};
  \draw[dotted] (A) -- (As) node[pos=0.8,above] {};
  \draw[-,very thick] (As) -- (Bs) node[pos=0.25,right] {} node[pos=0.75,right] {};
  \draw[-] (A) -- (Bs) node[pos=0.5,left] {};
  \draw[-] (B) -- (Bs) node[pos=0.8,right] {};

  % for M and Ms
  \draw[dashed] (M) -- (Bs) node[pos=0.5,above] {};
  \draw[dashed] (M) -- (As) node[pos=0.5,below] {};
  \draw[dashed] (M) -- (Ms) node[pos=0.8,above] {$d'$};

  % triangles
  \draw[schraffiert={green}{mypattern1}]{(A) -- (B) -- (Bs) -- (A)};
  \draw[schraffiert={green}{mypattern2}]{(A) -- (B) -- (As) -- (A)};
  \draw[schraffiert={blue}{mypattern3}]{(Bs) -- (M) -- (As) -- (Bs)};
  
\end{tikzpicture}  
\caption{Tetrahedron formed by extreme points $a$ and $b$ of agent $i$ and extreme points $a'$ and $b'$
  of agent $j$. The distance between the new agent positions $m$ and $m'$ is $d'$.}\label{fig:tetrahedron}
\end{figure}

Further, any three points among $a,b,a',b'$ must lie within a $2$ dimensional subspace, forming a triangle.
Lemma~\ref{lem:triangle} states the distance from the midpoint of two of its vertices to the opposite vertex, say $c$,
  and an upper bound in case the two edges incident to $c$ are upper bounded in length.

\begin{lem}\label{lem:triangle}
Let $\gamma \geq 0$ and $a,b,c \in V$.
Setting $m = (a+b)/2$,
we have
\[
\lVert m - c\rVert^2
=
\frac{1}{2}
\lVert a - c\rVert^2
+
\frac{1}{2}
\lVert b - c\rVert^2
-
\frac{1}{4}
\lVert a - b\rVert^2
\enspace.
\]
In particular, if
$\lVert a - c\rVert \leq \gamma$
and
$\lVert b - c\rVert \leq \gamma$,
then
\[
\lVert m - c\rVert^2
\leq
\gamma^2
-
\frac{1}{4}
\lVert a - b\rVert^2
\enspace.
\]
\end{lem}
\begin{proof}
We begin by calculating
\begin{equation}\label{eq:lem:triangle:1}
\begin{split}
\lVert a - c \rVert^2
& =
\big\lVert
(a-m) + (m - c)
\big\rVert^2
=
\lVert a - m \rVert^2
+
\lVert m - c \rVert^2
+
2\langle a-m, m-c \rangle
\end{split}
\end{equation}
and
\begin{equation}\label{eq:lem:triangle:2}
\begin{split}
\lVert b - c \rVert^2
& =
\big\lVert
(b-m) + (m - c)
\big\rVert^2
=
\lVert b - m \rVert^2
+
\lVert m - c \rVert^2
+
2\langle b-m, m-c \rangle
\enspace.
\end{split}
\end{equation}
Adding~\eqref{eq:lem:triangle:1} and~\eqref{eq:lem:triangle:2},
while noting $\lVert a - m\rVert^2 = \lVert b - m\rVert^2 = 
\frac{1}{4}\lVert a - b\rVert^2$
and $a-m = (a-b)/2 = -(b - m)$, gives
\begin{equation*}
\lVert a - c\rVert^2
+
\lVert b - c\rVert^2
=
\frac{1}{2}\lVert a - b \rVert^2
+
2\lVert m - c\rVert^2
\enspace.
\end{equation*}
Rearranging the terms in the last equation concludes the proof.
\end{proof}

We are now in the position to prove Lemma~\ref{lem:tetrahedron} that is central for Theorem~\ref{thm:midextremes}.
The lemma provides
  an upper bound on the distance $d'$ between $m$ and $m'$ for the tetrahedron
  in Figure~\ref{fig:tetrahedron} given that all
  its sides are upper bounded by some $\gamma \ge 0$ and the sum of the lengths of edge $a,b$ and $a',b'$, i.e.,
  $\lVert a-b \rVert + \lVert a'-b' \rVert$, is lower bounded by $\gamma$.
At the heart if the proof of Lemma~\ref{lem:tetrahedron} is an application of Lemma~\ref{lem:triangle} for the
  three hatched triangles in Figure~\ref{fig:tetrahedron}.
  
\begin{lem}\label{lem:tetrahedron}
Let $a,b,a',b'\in V$ and $\gamma\geq0$ such that
\begin{equation}\label{lem:tetrahedron:hyp}
\diam\!\big(\{a,b,a',b'\}\big) \leq \gamma \leq 
\lVert a - b\rVert + \lVert a' - b' \rVert
\enspace.
\end{equation}
Then, setting $m = (a+b)/2$ and $m' = (a'+b')/2$, we have
\[
\lVert m - m'\rVert
\leq
\sqrt{\frac{7}{8}}
\gamma
\enspace.
\]
\end{lem}
\begin{proof}
Applying Lemma~\ref{lem:triangle}
with the points $a,b,a'$ yields 
\begin{equation}\label{eq:t1}
\lVert m - a'\rVert^2 
\leq \gamma^2 - \frac{1}{4}\lVert a - b\rVert^2
\enspace.
\end{equation}
Another invocation with the points $a,b,b'$ gives
\begin{equation}\label{eq:t2}
\lVert m - b'\rVert^2 
\leq \gamma^2 - \frac{1}{4}\lVert a - b\rVert^2
\enspace.
\end{equation}
Now, again using Lemma~\ref{lem:triangle} with the points $a',b',m$ and
the bounds of~\eqref{eq:t1} and~\eqref{eq:t2}, we get
\begin{equation*}%\label{eq:t3}
\lVert m - m'\rVert^2 
\leq 
\gamma^2 
-
\frac{1}{4}
\left(
\lVert a - b\rVert^2
+
\lVert a'- b'\rVert^2
\right)
\enspace.
\end{equation*}
Using the second inequality in~\eqref{lem:tetrahedron:hyp}
then shows
\begin{equation}\label{eq:lem:tetrahedron:for:f}
\lVert m - m'\rVert^2
\leq
\gamma^2 
-
\frac{1}{4}
\left(
\lVert a - b\rVert^2
+
\big(\gamma - \lVert a - b\rVert\big)^2
\right)
\enspace.
\end{equation}
Setting $\xi = \lVert a - b\rVert$, we get
\begin{equation*}
\lVert m - m'\rVert^2
\leq
\max_{0\leq \xi\leq \gamma}
\gamma^2
-
\frac{1}{4}
\big(\xi^2 + (\gamma - \xi)^2\big)
\enspace.
\end{equation*}
Deriving the function $f(\xi) = \gamma^2
-
\frac{1}{4}
\big(\xi^2 + (\gamma - \xi)^2\big)$
reveals that its maximum is attained for
$-(2\xi-\gamma)=0$, i.e., $\xi = \gamma/2$, which gives
\begin{equation*}
\lVert m - m'\rVert^2
\leq
\gamma^2
-
\frac{\gamma^2}{8}
=
\frac{7}{8}\gamma^2
\enspace.
\end{equation*}
Taking the square root now concludes the proof.
\end{proof}

We can now prove Theorem~\ref{thm:midextremes}.
For the proof we consider the tetrahedron with vertices $a,b,a',b'$ as discussed before; see Figure~\ref{fig:tetrahedron}
Recalling that the vertices $a,b$ are vectors received by an agent $i$ and $a',b'$ vectors
  received by an agent $j$ in the same round, we may infer from the non-split property
  that all communication graphs must fulfill that both $i$ and $j$ must have received
  a common vector from an agent.
Together with the algorithm's rule of picking $a,b$ and $a',b'$ as extreme points, we obtain
  the constraints required by Lemma~\ref{lem:tetrahedron}.
Invoking this lemma we finally obtain an upper bound on the distance $d'$
  between $m$ in $m'$, and by this an upper bound on the round-by-round convergence rate 
  of the MidExtremes algorithm.
  
\begin{proof}[Proof of Theorem~\ref{thm:midextremes}]
Let $i$ and $j$ be two agents.
Let $a,b\in \Rcv_i(t)$ such that $y_i(t) = (a+b)/2$
and $a',b'\in \Rcv_j(t)$ such that $y_j(t) = (a'+b')/2$.
Define $\gamma_{ij} = \diam\big(\{a,b,a',b'\}\big)$.
Since $a,b,a',b'$ are the vectors of some agents in round $t-1$, we have
$\gamma_{ij} \leq \Delta\big(y(t-1)\big)$.

Further, from the non-split property, there is an agent~$k$ whose vector 
$c = y_k(t-1)$ has been received by both~$i$ and~$j$,
i.e., $c\in \Rcv_i(t) \cap \Rcv_j(t)$.
By the choice of the extreme points $a,b$ by agent $i$, we must have
$\lVert a - c \rVert \leq \lVert a - b \rVert$; otherwise $a,b$ would not
realize the diameter of $\Rcv_i(t)$.
Analogously, by the choice of the extreme points $a',b'$ by agent $j$, 
it must hold that $\lVert a' - c \rVert \leq \lVert a' - b'\rVert$.

From the triangular inequality, we then obtain 
\[
\lVert a - a' \rVert
\leq
\lVert a - c \rVert
+
\lVert c - a' \rVert
\leq
\lVert a - b \rVert
+
\lVert a'- b' \rVert
\enspace.
\]
Analogous arguments for the other pairs of points in $\{a,b,a',b'\}$ yield
\[
\diam\big(\{a,b,a',b'\}\big) = \gamma_{ij} \leq 
\lVert a - b \rVert
+
\lVert a'- b' \rVert
\enspace.
\]
We can hence apply Lemma~\ref{lem:tetrahedron} to obtain
\[
\lVert y_i(t) - y_j(t) \rVert 
\leq
\sqrt{\frac{7}{8}} \gamma_{ij}
\leq 
\sqrt{\frac{7}{8}} \Delta\big(y(t-1)\big)
\enspace.
\]
Taking the maximum over all pairs of agents~$i$ and~$j$ now shows
$\Delta\big(y(t)\big) \leq \sqrt{7/8} \Delta\big(y(t-1)\big)$,
which concludes the proof.
\end{proof}

\subsection{Bounds for ApproachExtreme}

We start by showing the one-dimensional case of Theorem~\ref{thm:approachextreme}, i.e., $V = \IR$,
  in Section~\ref{sec:onedimcase}.
Section~\ref{sec:multidimcase} then covers the multidimensional case.

\subsubsection{One-dimensional Case}
\label{sec:onedimcase}

For the proof we use the notion of $\varrho$-safety as introduced by 
Charron-Bost et al.~\cite{CBFN16}.
A convex combination algorithm is
{\em $\varrho$-safe\/}
if
\begin{equation}\label{eq:safety}
\varrho M_i(t) + (1-\varrho) m_i(t)
\leq
y_i(t)
\leq
(1-\varrho) M_i(t) + \varrho m_i(t)
\end{equation}
where 
$M_i(t) = \max \! \big(\!\Rcv_i(t)\big)$
and
$m_i(t) = \min \! \big(\!\Rcv_i(t)\big)$.

It was shown \cite[Theorem~4]{CBFN16} that any $\varrho$-safe convex combination
algorithm guarantees a round-by-round convergence rate of $c(t) \leq 1 - \varrho$ in any non-split
network model.
In the sequel, we will show that ApproachExtreme is $\frac{1}{4}$-safe
when applied in $V = \IR$.

\begin{proof}[Proof of Theorem~\ref{thm:approachextreme}, one-dimensional case]
Let~$i$ be an agent and $t\geq 1$ a round in some execution of ApproachExtreme
in $V = \IR$.
We distinguish the two cases
$y_i(t) \leq y_i(t-1)$
and
$y_i(t) > y_i(t-1)$.

In the first case, we have $b \leq y_i(t-1)$ for the vector~$b$ that agent~$i$ 
calculates in code line~4 in round~$t$.
But then necessarily $b = y_i(t)$ since this is the most distant point 
to $y_i(t-1)$ in
$\Rcv_i(t)$ to the left of $y_i(t-1)$.
Also, 
$y_i(t-1) \geq \big(M_i(t) + m_i(t)\big)/2$
since otherwise $M_i(t)$ would be farther from $y_i(t-1)$ than $m_i(t)$.
But this means that
\[
y_i(t)
=
\frac{y_i(t-1) + m_i(t)}{2}
\geq
\frac{1}{4} M_i(t) + \frac{1}{4} m_i(t) + \frac{1}{2} m_i(t)
=
\frac{1}{4}M_i(t) + \frac{3}{4}m_i(t)
\enspace,
\]
which shows the first inequality of $\varrho$-safety~\eqref{eq:safety}
with $\varrho = \frac{1}{4}$.
The second inequality of~\eqref{eq:safety} follows from $y_i(t-1) \leq M_i(t)$
since
\[
y_i(t)
=
\frac{y_i(t-1) + m_i(t)}{2}
\leq
\frac{1}{2} M_i(t) + \frac{1}{2} m_i(t)
\leq
\frac{3}{4} M_i(t) + \frac{1}{4} m_i(t)
\enspace.
\]

In the second case, \eqref{eq:safety} is proved analogously to the first
case.
\end{proof}

\subsubsection{Multidimensional Case}
\label{sec:multidimcase}

For the proof of Theorem~\ref{thm:approachextreme} with higher dimensional values, we
consider two agents $i, j$ whose distance realizes $\Delta(y(t))$.
From the ApproachExtreme $y_i(t) = m = (a + y_i(t-1))/2$ and $y_j(t) = m' = (a' + y_j(t-1))/2$ where
  $a$ and $a'$ maximize the distance to $y_i(t-1)$ and $y_j(t-1)$, respectively, among the received values.
%Figure~\ref{fig:tetrahedron2} depicts the tetrahedron formed by $y_i(t-1),a,y_j(t-1),a'$ within
%  its three-dimensional subspace.

%\begin{figure}[b]
%\begin{tikzpicture}[line join = round, line cap = round, scale=2]
%    
%  \coordinate[label=left:$a$] (A) at (0,0);
%  \coordinate[label=left:$y_i(t-1)$] (B) at (1,-1);
%  \coordinate[label=left:$m$] (M) at (0.5,-0.5);
%  
%  \coordinate[label=right:$a'$] (As) at (3,0);
%  \coordinate[label=right:$y_j(t-1)$] (Bs) at (2,1);
%  \coordinate[label=right:$m'$] (Ms) at (2.5,0.5);
%
%  \draw[-,very thick] (A) -- (B) node[pos=0.25,left] {} node[pos=0.75,left] {};
%  \draw[-] (B) -- (As) node[pos=0.5,right] {};
%  \draw[dotted] (A) -- (As) node[pos=0.8,above] {};
%  \draw[-,very thick] (As) -- (Bs) node[pos=0.25,right] {} node[pos=0.75,right] {};
%  \draw[-] (A) -- (Bs) node[pos=0.5,left] {};
%  \draw[-] (B) -- (Bs) node[pos=0.8,right] {};
%
%  % for M and Ms
%  \draw[dashed] (M) -- (Bs) node[pos=0.5,above] {};
%  \draw[dashed] (M) -- (As) node[pos=0.5,below] {};
%  \draw[dashed] (M) -- (Ms) node[pos=0.8,above] {$d'$};
%
%  % triangles
%  \draw[schraffiert={green}{mypattern1}]{(A) -- (B) -- (Bs) -- (A)};
%  \draw[schraffiert={green}{mypattern2}]{(A) -- (B) -- (As) -- (A)};
%  \draw[schraffiert={blue}{mypattern3}]{(Bs) -- (M) -- (As) -- (Bs)};
%  
%\end{tikzpicture}  
%\caption{Tetrahedron formed by the own points $y_i(t-1)$ and $y_j(t-1)$ and the extreme points $a$ and $a'$ of agent $i$ and
%  $j$, respectively. The distance between the new agent positions $m$ and $m'$ is $d'$.}\label{fig:tetrahedron2}
%\end{figure}

To show an upper bound on the distance $d'$ between the new agent positions~$m$ and~$m'$ in the multidimensional case,
we need the following variant of Lemma~\ref{lem:tetrahedron} in which we relax the upper bound
on~$\gamma$ by a factor of two, but thereby weaken the bound on $d'$.

Analogous to the proof of Theorem~\ref{thm:midextremes}, the proof is by applying
  Lemma~\ref{lem:tetrahedron:2} to the three hatched triangles in Figure~\ref{fig:tetrahedron}.

\begin{lem}\label{lem:tetrahedron:2}
Let $a,b,a',b'\in V$ and $\gamma\geq0$ such that
\begin{equation*}
\diam\!\big(\{a,b,a',b'\}\big) \leq \gamma \leq 
2\lVert a - b\rVert + 2\lVert a' - b' \rVert
\enspace.
\end{equation*}
Then, setting $m = (a+b)/2$ and $m' = (a'+b')/2$, we have
\[
\lVert m - m'\rVert
\leq
\sqrt{\frac{31}{32}}
\gamma
\enspace.
\]
\end{lem}

The proof of the lemma is essentially the same as that of 
Lemma~\ref{lem:tetrahedron}, with the following differences:
Equation~\eqref{eq:lem:tetrahedron:for:f} is replaced by
\begin{equation*}
\lVert m - m'\rVert^2
\leq
\gamma^2 
-
\frac{1}{4}
\left(
\lVert a - b\rVert^2
+
\left(\frac{\gamma}{2} - \lVert a - b\rVert\right)^2
\right)
\enspace,
\end{equation*}
which changes the function~$f$ to 
$f(\xi) = \gamma^2 - \frac{1}{4} \big( \xi^2 + (\frac{\gamma}{2} - \xi)^2
\big)$.
The maximum of this function~$f$ is achieved for $\xi = \gamma/4$, which 
means that
\begin{equation*}
\lVert m - m' \rVert^2
\leq
f(\gamma/4)
=
\gamma^2 - \frac{\gamma^2}{32}
=
\frac{31}{32} \gamma^2
\enspace.
\end{equation*}

\bigskip

We are now in the position to prove Theorem~\ref{thm:approachextreme}.

\begin{proof}[Proof of Theorem~\ref{thm:approachextreme}, multidimensional case]
Let~$i$ and~$j$ be two agents.
Let $a = y_i(t-1)$ and $a' = y_j(t-1)$.
Further, let $b\in \Rcv_i(t)$ such that $y_i(t) = (a + b)/2$
and $b'\in\Rcv_j(t)$ such that $y_j(t) = (a' + b')/2$.
Define $\gamma_{ij} = \diam\big(\{a,b,a',b'\}\big)$.
Since $a,b,a',b'$ are the vectors of some agents in round $t-1$, we have
$\gamma_{ij} \leq \Delta\big(y(t-1)\big)$.

From the non-split property, there is an agent~$k$ whose vector 
$c = y_k(t-1)$ has been received by both~$i$ and~$j$,
i.e., $c\in \Rcv_i(t) \cap \Rcv_j(t)$.
By the choice of the extreme point~$b$ by agent~$i$, we must have
$\lVert a - c \rVert \leq \lVert a - b \rVert$; otherwise~$b$ would not
maximize the distance to~$a$.
Analogously, by the choice of the extreme points~$b'$ by agent $j$, 
it must hold that $\lVert a' - c \rVert \leq \lVert a' - b'\rVert$.
Note, however, that the roles of~$a$ and~$b$ are not symmetric and that,
contrary to the proof of Theorem~\ref{thm:midextremes},
we can have $\lVert b - c \rVert > \lVert a - b \rVert$
or $\lVert b' - c \rVert > \lVert a' - b'\rVert$.

From the triangular inequality and the two established inequalities, we then
obtain 
\[
\lVert a - a' \rVert
\leq
\lVert a - c \rVert
+
\lVert a'- c \rVert
\leq
\lVert a - b \rVert
+
\lVert a'- b' \rVert
\enspace,
\]
\[
\lVert a - b' \rVert
\leq
\lVert a - c \rVert
+
\lVert c - a' \rVert
+
\lVert a' - b' \rVert
\leq
\lVert a - b \rVert
+
2\lVert a'- b' \rVert
\enspace,
\]
and
\[
\lVert b - b' \rVert
\leq
\lVert b - a \rVert
+
\lVert a - c \rVert
+
\lVert c - a' \rVert
+
\lVert a' - b' \rVert
\leq
2\lVert a - b \rVert
+
2\lVert a'- b' \rVert
\enspace.
\]
Analogously, $\lVert a' - b\rVert \leq 2\lVert a - b\rVert + \lVert a' -
b'\rVert$.
Together this implies
\[
\diam\big(\{a,b,a',b'\}\big) = \gamma_{ij} \leq 
2\lVert a - b \rVert
+
2\lVert a'- b' \rVert
\enspace.
\]
We can hence apply Lemma~\ref{lem:tetrahedron:2} to obtain
\[
\lVert y_i(t) - y_j(t) \rVert 
\leq
\sqrt{\frac{31}{32}} \gamma_{ij}
\leq 
\sqrt{\frac{31}{32}} \Delta\big(y(t-1)\big)
\enspace.
\]
Taking the maximum over all pairs of agents~$i$ and~$j$ now shows
$\Delta\big(y(t)\big) \leq \sqrt{31/32} \Delta\big(y(t-1)\big)$,
which concludes the proof.
\end{proof}

\section{Conclusion}\label{sec:conclusion}

We presented two new algorithms for asymptotic and approximate consensus
with values in arbitrary inner product spaces.
This includes not only the Euclidean spaces~$\IR^d$, but also spaces of 
infinite dimension.
Our algorithms are the first to have constant contraction rates, independent of
the dimension and the number of agents.

We have presented our algorithms in the framework of non-split network models
and have then shown how to apply them in several other distributed
computing models.
In particular, we improved the round complexity of the algorithms
by Mendes et al.~\cite{MHVG15} 
for asynchronous message passing with Byzantine faults
from 
$\Omega\big( d\log \frac{d\Delta}{\varepsilon} \big)$
to
$O\big( \! \log \frac{\Delta}{\varepsilon} \big)$,
eliminating all terms that depend on the dimension~$d$.

The exact value of the optimal convergence rate 
for asymptotic and approximate consensus
is known to be~$1/2$ in
dimensions one and two~\cite{FNS18:podc,CBFN16:centroid}, but the question
is still open for higher dimensions.
Our results are a step towards the solution of the problem as they show the
optimum in all dimensions to lie between~$1/2$
and $\sqrt{7/8} \approx 0.9354\dots$

\bibliography{agents}

\end{document}